\documentclass[runningheads]{llncs}
\usepackage[dvipdfmx]{graphicx}
\usepackage{comment}
\usepackage{amssymb}
\usepackage{amsfonts}
\usepackage{bm}
\usepackage{float}
\usepackage{amsmath}
\usepackage[linesnumbered,ruled,vlined]{algorithm2e}
\newtheorem{thm}{Theorem}
\newtheorem{dfn}[thm]{Definition}
\newtheorem{lem}[thm]{Lemma}

\newtheorem{prob}[thm]{Problem}
\newcounter{mycounter}
\renewcommand{\themycounter}{\arabic{mycounter}}
\newcommand{\mycounter}{\refstepcounter{mycounter}\themycounter}
\makeatletter
\renewcommand{\@Opargbegintheorem}[4]{%
  #4\trivlist\item[\hskip\labelsep{#3#2\@thmcounterend}]}
\makeatother
\begin{document}
\title{Dynamic Range Mode Enumeration}
\author{Tetto Obata}
\authorrunning{T. Obata}
\institute{Graduate School of Information Science and Technology, The University of Tokyo, Japan
\email{obata-tetto777@g.ecc.u-tokyo.ac.jp}}
\maketitle              % typeset the header of the contribution
\begin{abstract}
The range mode problem is a fundamental problem and there is a lot of work about it.
There is also some work for the dynamic version of it and the enumerating version of it, 
but there is no previous research about the dynamic and enumerating version of it.
We found an efficient algorithm for it.

\end{abstract}
\section{Introduction}
\begin{dfn}[mode]
    $A$ : multiset \\
    $a \in A$ is a mode of $A$\\
    $\Leftrightarrow$
    $\forall b \in A$ (the multiplicty of $a$ in $A$) $\geq$ (the multiplicty of $b$ in $A$)
\end{dfn}
In the following, ``a mode of multiset $\left\{A[l], A[l+1], \ldots, A[r]\right\}$'' is abbreviated to ``a mode of $A[l:r]$'' for a sequence $A$.
\begin{prob}[Range mode problem]
    Given a sequence $A$ over an alphabet set $\Sigma$,
    process a sequence of queries.
    \begin{itemize}
        \item {\rm mode}$\left(l, r\right)$: output one of the modes of A[l:r]
    \end{itemize}
\end{prob}
The range mode problem is a fundamental problem and there is a lot of work about it.
\begin{table}[H]
  \begin{center}
    \begin{tabular}{|c|c|c|c|} \hline
       & space complexity (bits) & query time complexity & conditions \\ \hline
       \cite{sta1} & ${\rm O}\!\left(n^{2-2\epsilon}\log n\right)$ & ${\rm O}\!\left(n^\epsilon\right)$ & $0\leq \epsilon \leq \frac{1}{2}$ \\ \hline
       \cite{durocher} & ${\rm O}\!\left(n^{2-2\epsilon}\right)$ & ${\rm O}\!\left(n^\epsilon\right)$ & $0\leq \epsilon \leq \frac{1}{2}$ \\ \hline
       \cite{sta3} & ${\rm O}\!\left(\frac{n^2\log \log n}{\log n}\right)$ & ${\rm O}\!\left(1\right)$ & \\ \hline
       \cite{sta4} & ${\rm O}\!\left(nm\log n\right)$ & ${\rm O}\!\left(\log m\right)$ & \\ \hline
       \cite{sta5} & ${\rm O}\!\left(\left(n^{1-\epsilon}m+n\right)\log n\right)$ & ${\rm O}\!\left(n^\epsilon + \log \log n\right)$ & $0\leq \epsilon \leq \frac{1}{2}$ \\ \hline
       \cite{Sumigawa} & ${\rm O}\!\left(4^knm\left(\frac{n}{m}\right)^{\frac{1}{2^{2^k}}}\right)$ & ${\rm O}\!\left(2^k\right)$ & $k \in \mathbb{Z}_{\geq 0}$ \\ \hline
       \cite{Sumigawa} & ${\rm O}\!\left(nm\right)$ & ${\rm O}\!\left(\min\left(\log m, \log \log n\right)\right)$ &  \\ \hline
       \cite{Sumigawa} & ${\rm O}\!\left(nm\left(\log \log\frac{n}{m}\right)^2\right)$ & ${\rm O}\!\left(\log \log \frac{n}{m}\right)$ &  \\ \hline
    \end{tabular}
    \caption{The results of previous research about the range mode problem. $n$ is the length of a string and $m$ is the maximum frequency of an item. Space complexity does not include the input string.}
    \label{Static}
  \end{center}
\end{table}
As a natural extension of the range mode problem, we can consider the enumeration version of the problem.
\begin{prob}[Range mode enumeration problem]
    Given a sequence $A$ over an alphabet set $\Sigma$,
    process a sequence of queries.
    \begin{itemize}
        \item {\rm modes}$\left(l, r\right)$: enumerate the modes of A[l:r]
    \end{itemize}
\end{prob}
There is another natural extension of it, the dynamic version of the problem.
\begin{prob}[Dynamic range mode problem]
    Given a sequence $A$ over an alphabet set $\Sigma$, process a sequence of queries of the following three types:
    \begin{itemize}
        \item {\rm insert}$\left(c, i\right)$: insert $c \left(\in \Sigma\right)$ so that it becomes the $i$-th element of $A$
        \item {\rm delete}$\left(i\right)$: delete the $i$-th element of $A$
        \item {\rm mode}$\left(l, r\right)$: output one of the modes of A[l:r]
    \end{itemize}
\end{prob}
There is some work about the range mode enumeration problem and the dynamic range mode problem.
\begin{table}[htbp]
  \begin{center}
    \begin{tabular}{|c|c|c|c|} \hline
       & space complexity (bits) & query time complexity & condition \\ \hline
      \cite{Sumigawa} & ${\rm O}\!\left(n^{2-2\epsilon}\log n\right)$ & ${\rm O}\!\left(n^\epsilon\left|output\right|\right)$ & $0 \leq \epsilon \leq \frac{1}{2}$ \\ \hline
      \cite{Sumigawa} & ${\rm O}\!\left(nm\left(\log \log \frac{n}{m}\right)^2 + n\log n\right)$ & ${\rm O}\!\left(\log \log \frac{n}{m} + \left|output\right|\right)$ & \\ \hline
      \cite{Sumigawa} & ${\rm O}\!\left(nm + n\log n\right)$ & ${\rm O}\!\left(\log m + \left|output\right|\right)$ & \\ \hline
      \cite{Sumigawa} & ${\rm O}\!\left(n^{1+\epsilon}\log n + n^{2-\epsilon}\right)$ & ${\rm O}\!\left(\log m + n^{1-\epsilon} + \left|output\right|\right)$ & $0 \leq \epsilon \leq 1$ \\ \hline
    \end{tabular}
    \caption{The results of previous research about the range mode enumeration problem. $n$ is the length of a string and $m$ is the maximum frequency of an item. Space complexity does not include the input string.}
    \label{Enumeration}
  \end{center}
\end{table}
\begin{table}[H]
  \begin{center}
    \begin{tabular}{|c|c|c|c|} \hline
       & space complexity(words) & query time complexity\\ \hline
      \cite{dynamic1} & ${\rm O}\!\left(n_{\max} \right)$ & ${\rm O}\!\left(n_{\max}^{\frac{2}{3}}\right)$ \\ \hline
      \cite{dynamic2} & ${\rm \tilde{O}}\!\left(n^{1.327997}\right)$ & ${\rm \tilde{O}}\!\left(n^{0.655994}\right)$ \\ \hline
    \end{tabular}
    \caption{The results of previous research about the dynamic range mode problem where $n$ is the length of string and $n_{\max}$ is the limit of the length of the string. Space complexity does not include the input string. The query time comlexity is same in all the query types. The wordsize is ${\rm \Omega}\! \left(\log n\right)$.}
    \label{Dynamic}
  \end{center}
\end{table}
Considering the normal version, enumerating version, and the dynamic version of the problem, 
we can consider another problem, the dynamic enumerating version one.
\begin{prob}[Dynamic range mode enumeration problem]
    Given a sequence $A$ over an alphabet set $\Sigma$, process a sequence of queries of the following three types:
    \begin{itemize}
        \item {\rm insert}$\left(c, i\right)$: insert $c \left(\in \Sigma\right)$ so that it becomes the $i$-th element of $A$
        \item {\rm delete}$\left(i\right)$: delete the $i$-th element of $A$
        \item {\rm modes}$\left(l, r\right)$: enumerate the modes of A[l:r]
    \end{itemize}
\end{prob}
There is no previous research about the dynamic range mode enumeration problem.\\
It is known that the range mode problem is related to the boolean matrix problem and the set intersection problem~\cite{sta5}.
\begin{prob}[Set intersection problem]
    Given multisets $S_1, S_2, \ldots, S_N$ of a universe $U$, process a sequence of following queries.
    \begin{itemize}
        \item {\rm intersect}$\left(i, j\right)$: check whether $S_i$ and $S_j$ intersect or not
    \end{itemize}
\end{prob}
If the range mode problem can be solved efficiently, it can be checked if two sets intersect efficiently.
We can solve the set intersection problem by building a data structure for a sequence of $2N\left|U\right|$ elements as follows
\begin{align*}
    \mbox{(elements of) } S_1, S_1^{\rm c}, S_1^{\rm c}, S_1, 
    S_2, S_2^{\rm c}, S_2^{\rm c}, S_2, 
    \ldots , 
    S_N, S_N^{\rm c}, S_N^{\rm c}, S_N
\end{align*}
and calling mode$\left(2i\left|U\right| - \left|S_i\right|, 2(j-1)\left|U\right| + \left|S_j\right|\right)$ query for a intersect$\left(i, j\right)$ $\left(i < j\right)$ query.\\
Therefore if the dynamic range mode enumeration problem can be solved efficiently, the computation of the intersection of two sets and modifying of the sets can be done efficiently.
\subsection*{Our contribution}
Existing methods for the dynamic range mode problem cannot be applied to the dynamic range mode enumeration problem.
The step 3 of Algorithm 1 of \cite{dynamic1} cannot be used for the dynamic range mode enumeration problem.
Problem 7 of \cite{dynamic2} needs only one index and the algorithm of this paper is based on this problem.
In this paper, we found the first algorithm for the range dynamic enumeration problem, which can deal with insert and delete queries in
${\rm O}\!\left(N^{\frac{2}{3}}\log \sigma^\prime \right)$ time per query and modes query in ${\rm O}\!\left(N^{\frac{2}{3}}\log \sigma^\prime + |output| \right)$
    time per query where $N$ is the length of the sequence and $\sigma^\prime = \left|\left\{c \in \Sigma \middle| c \mbox{ appears in the sequence}\right\}\right|$.
\section{Main Result}
%\subsection{content}
The following theorem is the main result.
\begin{thm}
    \label{mainresult}
    There exists a data structure for the dynamic range mode enumeration problem in the word RAM model with ${\rm \Omega}\! \left(\log N + \log \sigma\right)$ bits wordsize
    in ${\rm O}\!\left(N^{\frac{2}{3}}\log \sigma^\prime \right)$ time per {\rm insert} and {\rm delete} query and ${\rm O}\!\left(N^{\frac{2}{3}}\log \sigma^\prime + |output| \right)$ 
    time per {\rm modes} query where $N$ is the length of the sequence and \\$\sigma^\prime = \left|\left\{c \in \Sigma \middle| c \mbox{ appears in the sequence}\right\}\right|$.
    The space complexity is ${\rm O}\!\left(N  + N^{\frac{2}{3}}\sigma^\prime\right)$ words.
\end{thm}

Our main idea is to divide the sequence into $L = {\rm \Theta}\! \left(N^\alpha\right)$ subsequences of length which may be zero but not greater than $C = {\rm \Theta}\!\left(N^{1-\alpha}\right)$ for some parameter $a$.
Let $B_i$ be the $i$-th subsequence. We call it a block. For sequences $X, Y$, we define $X + Y$ as the sequence obtained by concatenating $X$ and $Y$ in this order.\\
The data structure consists of the following components.
\begin{itemize}
    \item $T_A$ : A data structure for the sequence $A$. It can process the following queries.
    \begin{itemize}
        \item access $A\left[l:r\right]$ $\left(0 \leq l \leq r < \left|A\right|\right)$
        in ${\rm O}\!\left(T_{\mycounter\label{accessA}, r-l+1}\right)$ time.
        \item insert a character $c\left(\in \Sigma\right)$ into $i$-th position of $A$ $\left(0 \leq i \leq \left|A\right|\right)$
        in ${\rm O}\!\left(T_{\mycounter\label{insertA}}\right)$ time.
        \item delete the $i$-th character of $X$ $\left(0 \leq i < \left|A\right|\right)$
        in ${\rm O}\!\left(T_{\mycounter\label{deleteA}}\right)$ time.
    \end{itemize}
    \item $T_B$ : A data structure for the array $\left(\left|B_0\right|, \left|B_1\right|, \ldots, \left|B_{L-1}\right|\right)$, which is used to compute which block a character in $A$ belongs to. It can process the following queries.
    \begin{itemize}
        \item increase or decrease the $i$-th element $\left(0 \leq i < L\right)$
        in ${\rm O}\!\left(T_{\mycounter\label{modifyB}}\right)$ time.
        \item calculate $\mathrm{argmin}_i \left|B_i\right|$
        in ${\rm O}\!\left(T_{\mycounter\label{minB}}\right)$ time.
        \item calculate $\min \left\{k \middle| \sum_{i = 0}^{k} \left|B_i\right| \geq a \right\}$ $\left(0 < a \leq \sum_i \left|B_i\right| \right)$
        in ${\rm O}\!\left(T_{\mycounter\label{binaryB}}\right)$ time.
        \item insert a value $x$ into $i$-th position of the array
        in ${\rm O}\!\left(T_{\mycounter\label{insertB}}\right)$ time.
        \item delete the $i$-th element of the array
        in ${\rm O}\!\left(T_{\mycounter\label{deleteB}}\right)$ time.
    \end{itemize}
    \item $S_{(l,r)} \left(0 \leq l \leq r < L\right)$ : A data structure for the ordered set\\$\left\{\left(\left(\mbox{the multiplicity of } c \mbox{ in } B_l + \cdots + B_r\right) , c\right) \middle| c \mbox{ appears in } B_l + \cdots + B_r\right\}$.
    It can process the following queries.
    \begin{itemize}
        \item create an empty set
        in ${\rm O}\!\left(T_{\mycounter\label{initS}}\right)$ time.
        \item increment or decrement the multiplicity of character $c\left(\in \Sigma\right)$
        in ${\rm O}\!\left(T_{\mycounter\label{modifyS}}\right)$ time.
        \item compute the multiplicity of a character $c \left(\in \Sigma\right)$
        in ${\rm O}\!\left(T_{\mycounter\label{accessS}}\right)$ time.
        \item access the largest element
        in ${\rm O}\!\left(T_{\mycounter\label{topS}}\right)$ time.
        \item access the next largest to the last accessed element
        in ${\rm O}\!\left(T_{\mycounter\label{nextS}}\right)$ time.
    \end{itemize}
\end{itemize}

We introduce new operations ${\rm moveLeft}\left(i\right)$ and ${\rm moveRight}\left(i\right)$.
The operation ${\rm moveLeft}\left(i\right)$ moves the first element of $i$-th block to the $(i-1)$-st block.
In such an operation, we only need to modify the following components.
\begin{itemize}
    \item $T_B$
    \item $S_{(0,i-1)}, \ldots, S_{(i-1,i-1)}, S_{(i,i)}, \ldots, S_{(i,L-1)}$
\end{itemize}
This can be done in ${\rm O}\!\left(T_{\ref{modifyB}} + LT_{\ref{modifyS}}\right)$ time.
The operation ${\rm moveRight}\left(i \right)$ moves the last element of $i$-th block to the $(i+1)$-st block.
It can be done in the same time in a similar way.

We process the queries by the following method.

\subsubsection*{delete}
Let $j$ be the index of the block that contains the $i$-th element.
It can be computed in ${\rm O}\!\left(T_{\ref{binaryB}} \right)$ time.
$T_A$ and $T_B$ can be modified easily in ${\rm O}\!\left(T_{\ref{accessA}, 1} + T_{\ref{deleteA}} + T_{\ref{modifyB}}\right)$ time.
We need to modify $S_{(l,r)}$ for all $l, r$ such that $0 \leq l \leq j \leq r < L$. It can be done in ${\rm O}\!\left(L^2 T_{\ref{modifyS}} \right)$ time.

\subsubsection*{insert}
Let $j$ be $\min \left\{k \middle| \sum_{l = 0}^{k} |B_l| \geq i \right\}$. We insert $c$ into the $j$-th block, 
and modify the data structure in a similar way to a delete query.\\
The length of $j$-th block may become larger than $C$. In such a case we balance the length of blocks in the following way.
\begin{enumerate}
    \item Find a block $B_k$ such that $|B_k| + 1 \leq C$.
    \item Operate ${\rm moveLeft}$ or ${\rm moveRight}$ several times so that $|B_j|$ decreases by 1, $|B_k|$ increases by 1 and the rest remain.
\end{enumerate}
Step 1. can be done in ${\rm O}\!\left(T_{\ref{minB}}\right)$ time. Step 2. can be done in ${\rm O}\!\left(L\left(T_{\ref{modifyB}} + LT_{\ref{modifyS}}\right)\right)$ time
because we call {\rm moveLeft} or {\rm moveRight} only ${\rm O}\!\left(L\right)$ times.\\
\subsubsection*{modes}
Let $(i, j)$ be the maximal interval of blocks which is in $A[l:r]$.
It can be computed in ${\rm O}\!\left(T_{\ref{binaryB}}\right)$ time.
If $A[l:r]$ does not contain any blocks, $B_i + \cdots + B_j$ stands for an empty sequence and $S_{l, r}$ stands for an empty set below.\\
It holds that $|A[l:r] \setminus \left(B_i + \cdots + B_j\right)| \leq 2 C$.
It can be said that every mode of $A[l:r]$ is a mode of $B_i + \cdots + B_j$ or appears in 
$A[l:r] \setminus \left(B_i + \cdots + B_j\right)$.
If there does not exists such a character $c$ that meets the following conditions
\begin{itemize}
    \item $c$ is a mode of $A[l:r]$
    \item $c$ does not appear in $A[l:r] \setminus \left(B_i + \cdots + B_j\right)$
\end{itemize}
then every mode of $A[l:r]$ appears in $A[l:r] \setminus \left(B_i + \cdots + B_j\right)$.
We scan the elements in $A[l:r] \setminus \left(B_i + \cdots + B_j\right)$
and count the occurrences of each character in $A[l:r]$ using $S_{\left(i,j\right)}$ and a new ordered set in ${\rm O}\!\left(T_{\ref{initS}} + CT_{\ref{modifyS}} + T_{\ref{accessA}, C} + CT_{\ref{accessS}} \right)$ time, 
and compute the number of occurrences of a mode of $A[l:r]$ in ${\rm O}\!\left(T_{\ref{accessA}, C} + CT_{\ref{accessS}} + T_{\ref{topS}}\right)$ time.
We can judge if there exists a character satisfying the conditions above using the value and the ordered set.
If there does not exist such a character, the enumeration is done.
If exists, every mode of $B_i + \cdots + B_j$ is also a mode of $A[l:r]$, 
so we can enumerate the modes of $A[l:r]$ by the privious scan and the enumeration of the modes of $B_i + \cdots + B_r$, which can be done in ${\rm O}\!\left(T_{\ref{topS}} + T_{\ref{nextS}}\left|output\right|\right)$ time.
Algorithm \ref{modes} denotes the algorithm for the modes query.\\
\begin{algorithm}[H]
    \caption{The algorithm for the modes query.}         
    \label{modes}                          
\SetKwInput{KwInput}{Input}               % Set the Input
\SetKwInput{KwOutput}{Output}              % set the Output
\DontPrintSemicolon

    \KwInput{range $\left(l, r\right)$}
    \KwOutput{all modes of $S[l:r]$}
  
    \SetKwFunction{FMain}{Main}
  
    \SetKwProg{Fn}{Function}{:}{\KwRet}
    \Fn{\FMain}{
        $\left(i, j\right) \leftarrow$ maximal interval such that $B_i + \cdots + B_j \subset A[l:r]$\;
        $T \leftarrow$ a new empty ordered set\;
        \For{$c \in A[l:r] \setminus \left(B_i + \cdots B_j\right)$}{
            increment the multiplicity of $c$ in $T$\;
        }
        $app \leftarrow 0$\;
        \For{$\left(app_c, c\right) \in T$}{
            $app = \max\left(app, app_c + \left(\mbox{the number of appearences of }c \mbox{ in } B_i + \cdots B_j\right) \right)$\;
        }
        $ans \leftarrow \varnothing$\;
        \For{$\left(app_c, c\right) \in T$}{
            \If{$app = app_c + \left(\mbox{the number of appearences of }c \mbox{ in } B_i + \cdots B_j\right)$}
            {
                $ans \leftarrow ans \cup \left\{c\right\}$\;
            }
        }
        $\left(app_c, c\right) \leftarrow$ the top element of $S_{\left(i, j\right)}$\;
        \While{$c \neq NULL$ and $app_c = app$}{
            $ans \leftarrow ans \cup \left\{c\right\}$\;
            $\left(app_c, c\right) \leftarrow$ the next largest to $\left(app_c, c\right)$ in $S_{\left(i, j\right)}$\;
        }
        \KwRet $ans$\;
    }
\end{algorithm}
In order to keep $L = {\rm \Theta}\!\left(N^\alpha\right)$ and $C = {\rm \Theta}\!\left(N^{1-\alpha}\right)$, we use the technique for dynamic data structures~\cite{navarro}.
We group the blocks into three types $\rm{p, c, n}$(previous, current, next).
Set the number and size of the blocks as follows.
\begin{itemize}
    \item $L_{\rm p} = \lceil\left(\frac{N}{2}\right)^\alpha\rceil, C_{\rm p} = \lceil\left(\frac{N}{2}\right)^{1-\alpha}\rceil$
    \item $L_{\rm c} = \lceil N^\alpha\rceil, C_{\rm c} = \lceil N^{1-\alpha}\rceil$
    \item $L_{\rm n} = \lceil \left(2N\right)^\alpha\rceil, C_{\rm n} = \lceil \left(2N\right)^{1-\alpha}\rceil$
\end{itemize}
When we initialize the data structure, all elements are stored in c blocks and initialize the data structure for $L = L_{\rm p} + L_{\rm c} + L_{\rm n}$ blocks.
\subsubsection*{insert}
Move elements so that the sum of the elements in n blocks increases by two and that in p blocks decreases by one (unless they are already empty) compared to before the query.
To achieve this, we move elements as follows
\begin{itemize}
    \item insertion into p: p $\rightarrow$ c, p $\rightarrow$ c, c $\rightarrow$ n, c $\rightarrow$ n
    \item insertion into c: p $\rightarrow$ c, c $\rightarrow$ n, c $\rightarrow$ n
    \item insertion into n: p $\rightarrow$ c, c $\rightarrow$ n
\end{itemize}
where x $\rightarrow$ y means moving the last element of x blocks to y blocks.
If all x blocks are empty, it is ignored.
\subsubsection*{delete}
Move elements so that the sum of the elements in n blocks decreases by two (unless they are already empty) and that in p blocks increases by one compared to before the query.
To achieve this, we move elements as follow
\begin{itemize}
    \item insertion into p: c $\leftarrow$ n, c $\leftarrow$ n, p $\leftarrow$ c, p $\leftarrow$ c
    \item insertion into c: c $\leftarrow$ n, c $\leftarrow$ n, p $\leftarrow$ c
    \item insertion into n: c $\leftarrow$ n, p $\leftarrow$ c
\end{itemize}
where x $\leftarrow$ y means moving the first element of y blocks to x blocks.
If all y blocks are empty, it is ignored.
\begin{lem}{\cite{navarro}}
    When the length of the string becomes double, all elements are in {\rm n} blocks.\\
    When the length of the string becomes half, all elements are in {\rm p} blocks.
\end{lem}
    If the length of the string becomes double, set the blocks as follows
    \begin{align*}
        \left(\rm{p, c, n}\right) \leftarrow \left(\rm{pp, p, c}\right)
    \end{align*}
    and if the length of the string becomes half, set the blocks as follows
    \begin{align*}
        \left(\rm{p, c, n}\right) \leftarrow \left(\rm{c, n, nn}\right)
    \end{align*}
    where pp (previous to the previous) and nn (next to the next) are other types of blocks and $L_{\rm pp}, C_{\rm pp}, L_{\rm nn}, $ and $C_{\rm nn}$ are defined as follow.
    \begin{itemize}
        \item $L_{\rm pp} = \lceil\left(\frac{N}{4}\right)^\alpha\rceil, C_{\rm pp} = \lceil\left(\frac{N}{4}\right)^{1-\alpha}\rceil$
        \item $L_{\rm nn} = \lceil \left(4N\right)^\alpha\rceil, C_{\rm nn} = \lceil\left(4N\right)^{1-\alpha}\rceil$
    \end{itemize}
    We need to add ${\rm O}\!\left(N^{1 + 2\alpha}\right)$ extra elements for $S_{\left(l, r\right)} (0 \leq l \leq r < L_{\rm pp} + L_{\rm p} + L_{\rm c}$ $+ L_{\rm n} + L_{\rm nn})$ in order to prepare pp blocks and nn blocks and are prepared from when the block reset occured. There are ${\rm \Omega}\! \left(N\right)$ queries.
    These operations need\\
    %It can be done in ${\rm O}\!\left(L^2 T_{\ref{modifyS}} \right)$ time.
    ${\rm O}\!\left(T_{\ref{minB}} + T_{\ref{insertB}} + T_{\ref{deleteB}}
    + L\left(T_{\ref{modifyB}} + LT_{\ref{modifyS}}\right) 
    + N^{2\alpha}T_{\ref{modifyS}} + \max\left(1, N^{2\alpha - 1}\right)T_{\ref{initS}} \right)$ time per query.
\begin{thm}
    There exists a data structure for the dynamic range mode \mbox{enumeration} problem 
    %in the word RAM model with ${\rm \Omega}\! \left(\log N + \log \sigma\right)$ bits wordsize
    in
    ${\rm O}\!\left(T_{\ref{accessA}, 1} + T_{\ref{insertA}} + T_{\ref{deleteA}} + T_{\ref{minB}} + T_{\ref{binaryB}}
    + N^{\alpha}T_{\ref{modifyB}} + N^{2\alpha}T_{\ref{modifyS}}
    + \max\left(1, N^{2\alpha - 1}\right)T_{\ref{initS}} \right)$
    time per {\rm insert} and {\rm delete} query and
    ${\rm O}\!\left(T_{\ref{binaryB}}
    + T_{\ref{accessA}, {\rm \Theta}\!\left(N^{1-\alpha}\right)} + N^{1-\alpha}T_{\ref{accessS}} + T_{\ref{topS}}
    + T_{\ref{nextS}}\left|output\right|\right)$
    time per {\rm modes} query where $N$ is the length of the sequence
    and $\sigma = \left|\Sigma\right|$.
\end{thm}
%\begin{refproof}[proof of Theorem \ref{mainresult}.]
\begin{proof}[Proof of Theorem \ref{mainresult}]
    We use balanced binary search trees for $T_A$ and $T_B$.\\
    We use two balanced binary search trees for each $S_{\left(l, r\right)}$.
    One of them is the one whose key is a character in $B_l + \cdots + B_r$ and value is the number of the occurrences of the character in $B_l + \cdots + B_r$.
    The other is used as a ordered set $\left\{\left(t\left(c\right), c\right) \middle| c\in \Sigma, t\left(c\right) > 0\right\}$, where $t\left(c\right)$ is the number of the occurrences of $c$ in $B_l + \cdots + B_r$.
    Then, following equations hold.
    \begin{align*}
        &T_{\ref{accessA}, a} = {\rm O}\!\left(a + \log N\right)\\
        &T_{\ref{insertA}} = {\rm O}\!\left(\log N\right)\\
        &T_{\ref{deleteA}} = {\rm O}\!\left(\log N\right)\\
        &T_{\ref{modifyB}} = {\rm O}\!\left(\log L\right) = {\rm O}\!\left(\log N\right)\\
        &T_{\ref{minB}} = {\rm O}\!\left(\log L\right) = {\rm O}\!\left(\log N\right)\\
        &T_{\ref{binaryB}} = {\rm O}\!\left(\log L\right) = {\rm O}\!\left(\log N\right)\\
        &T_{\ref{insertB}} = {\rm O}\!\left(\log L\right) = {\rm O}\!\left(\log N\right)\\
        &T_{\ref{deleteB}} = {\rm O}\!\left(\log L\right) = {\rm O}\!\left(\log N\right)\\
        &T_{\ref{initS}} = {\rm O}\!\left(1\right)\\
        &T_{\ref{modifyS}} = {\rm O}\!\left(\log \sigma^\prime\right)\\
        &T_{\ref{accessS}} = {\rm O}\!\left(\log \sigma^\prime\right)\\
        &T_{\ref{topS}} = {\rm O}\!\left(1\right)\\
        &T_{\ref{nextS}} = {\rm O}\!\left(1\right)
    \end{align*}
    Setting $\alpha = \frac{1}{3}$, we obtain theorem \ref{mainresult}.
\end{proof}

\section{Concluding Remarks}
We introduced a new problem, the dynamic range mode enumeration problem.
We found an algorithm for it whose time complexity of a modes query is linear to the output size plus some term.
However, the term is larger than the time complexity of a mode query of the dynamic range mode problem.
It may be possible to found a new algorithm for the dynamic range mode enumeration problem whose time complexity for a query is equal to that of the dynamic range mode problem except the term depending on the output size.
\bibliographystyle{splncs04}
\bibliography{CITE.bib}
\end{document}